\documentclass[11pt,a4paper,english]{article}

\usepackage[margin=1.1in]{geometry}

\usepackage{microtype} 

\usepackage{graphicx}
\graphicspath{{./figures/}}

\usepackage{xspace}
\usepackage[utf8]{inputenc}
\usepackage{amssymb,amsfonts,amsmath,amsthm}
\usepackage[inline]{enumitem}
\usepackage{mathtools}
\usepackage[hidelinks]{hyperref}
\usepackage{csquotes}
\usepackage{thmtools, thm-restate}
\usepackage{xcolor}

\newtheorem{theorem}{Theorem}
\newtheorem{corollary}[theorem]{Corollary}
\newtheorem{lemma}[theorem]{Lemma}

\newtheorem{question}[theorem]{Question}

\usepackage[nameinlink,capitalise]{cleveref}
\Crefname{figure}{Figure}{Figures} 

\def\RR{{\mathbb{R}}}

\let\geq\geqslant
\let\leq\leqslant
\let\ge\geqslant
\let\le\leqslant

\title{Connected Matchings\thanks{A preliminary version of this work was presented at 
the 40th European Workshop on Computational Geometry (EuroCG2024).}}

\author{%
Oswin~Aichholzer\thanks{Institute of Algorithms and Theory, Graz University of Technology, Graz, Austria. \texttt{oswin.aichholzer@tugraz.at}.}
\and 
Sergio~Cabello\thanks{Faculty of Mathematics and Physics, University of Ljubljana, Slovenia, 
and Institute of Mathematics, Physics and Mechanics, Slovenia.
\texttt{sergio.cabello@fmf.uni-lj.si}.}
\and 
Viola~M{\'e}sz{\'a}ros\thanks{Bolyai Institute, University of Szeged, Szeged, Hungary. \texttt{meszaros.viola@gmail.com}} 
\and
Patrick~Schnider\thanks{Department of Computer Science, ETH Z\"urich, Z\"urich, Switzerland. \texttt{patrick.schnider@inf.ethz.ch}.}
\and 
Jan~Soukup\thanks{Faculty of Mathematics and Physics, Charles University, Prague, Czech Republic. \texttt{soukup@kam.mff.cuni.cz}} 
}

\begin{document}
\maketitle
	
\begin{abstract}
	We show that each set of $n\ge 2$ points in the plane in general position has 
	a straight-line matching with at least $(5n+1)/27$ edges whose segments form a connected set,
	and such a matching can be computed in $O(n \log n)$ time.
	As an upper bound, we show that for some planar point sets in general position
	the largest matching whose segments form a connected set has $\lceil \frac{n-1}{3}\rceil$ edges. 
	We also consider a colored version, where each edge of the matching
	should connect points with different colors.

\end{abstract}

\section{Introduction}
\label{sec:intro}

Consider a set $P$ of $n$ points in the plane in general position, meaning that no three points of $P$ are collinear. A (straight line) \emph{matching} $M$ for $P$ is a set of segments with endpoints in $P$ such that no two segments share an endpoint. A matching $M$ for $P$ is \emph{connected} (via their crossings) if the union of the segments of $M$ forms a connected set. Equivalently, a matching is connected when the intersection graph of its segments is connected. The \emph{size} of the matching $M$ is the number of segments (or edges) in $M$. Note that whenever $P$ has a connected matching of size $m\ge 1$, it also has a connected matching of size $m-1$. Indeed, this is easy to see using the formulation via  intersection graphs: in a connected graph, which is the intersection graph of the $m$ segments of a matching $M$, we can always remove a vertex (which is an edge in $M$), and keep the graph connected.

In this paper, we study the following problem.

\begin{question}[Connected Matching] \label{que:main}
	Find for each $n$ the largest value~$f(n)$ with the following property: 
	each set of $n$ points in general position in the plane 
	has a connected matching of size~$f(n)$.
\end{question}

It is also natural to consider a colored version of the problem. In this setting, the points are colored and each edge of the matching has to connect points with different colors. A \emph{balanced $c$-coloring} of $P$ is a partition of $P$ into $c$ subsets $P_1,\dots, P_c$ such that $|P_i|$ and $|P_j|$ differ by at most one, for each $1\le i,j\le c$. In particular, if $n$ is divisible by $c$, each set $P_i$ has cardinality $n/c$. A matching for a balanced $c$-coloring $P_1,\dots, P_c$ is \emph{polychromatic} if each segment connects two points with different colors.
The bichromatic version of the problem corresponds to $c=2$.
We are also interested in the following question.
	
\begin{question}[Colored Connected Matching] 
	Find for each $c\le n$ the largest value~$g(n,c)$ with the following property:
	each set of $n$ points in general position in the plane with a balanced $c$-coloring  
	has a connected polychromatic matching of size~$g(n,c)$.
\end{question}
	
The setting with $c=n$ colors corresponds to the uncolored version because all edges are allowed in the matching. Therefore, $f(n)=g(n,n)$.

In this work we provide upper and lower bounds for the functions $f(n)$
and $g(n,c)$. 
We show that $\frac{5n+1}{27}\le f(n)$ and a connected matching of this size can
be computed in $O(n \log n)$ time. We also show that $f(n)\le \lceil\frac{n-1}{3}\rceil$.
For the function $g(n,c)$, we provide an upper bound only in the bichromatic setting,
namely $g(n,2)\le \lceil\frac{n-1}{4}\rceil$.
We also show that, for sufficiently large $n$,
\[
	g(n,c)\ge \begin{cases}
		\displaystyle\frac{c-3}{6c}n-\frac{1}{2} & \text{for $c> 7$,}\\[.3cm]
		\displaystyle\frac{c-1}{9c}n-\frac{1}{3} & \text{ for $2\le c\le 7$.}
	\end{cases}
\]
For the bichromatic case, $c=2$, this bound gives $g(n,2)\ge \frac{n}{18}-\frac{1}{3}$.
When $c$ is very large, the lower bound becomes $\frac{n}{6}-\Theta(1)$.
Again, connected polychromatic matchings attaining this size can
be computed efficiently, namely in linear time.

The problem can be seen as a relaxation of the problem of \emph{crossing families}
of Aronov et al.~\cite{AronovEGKKPS94}, where one wants to find as many segments
as possible with endpoints in $P$ such that any pair of segments crosses 
in their interior. While in our setting we are asking
for a connected subgraph in the intersection graph of the segments, 
the crossing families problem asks that the intersection graph is
a complete graph. 
The best lower bound, showing an almost linear lower bound for crossing families,
has been a recent breakthrough by Pach, Rubin and Tardos~\cite{PachRT21}. 
Aichholzer et al.~\cite{AichholzerKSVV22} have the currently best upper bound.

The rest of the paper is organized as follows.
In \cref{sec:algo_tools} we provide some basic subroutines that will be used
in our algorithms.
In \cref{sec:separation} we discuss the existence and computation of a separator
for points in the plane; the existence of such a separator is discussed by \'Abrego and 
Fern\'andez-Merchant~\cite{AbregoF17}.
In \cref{sec:upper} we provide upper bounds for $f(n)$ and $g(n,c)$.
In \cref{sec:mono} we present a lower bound for $f(n)$  
and in \cref{sec:deep} we show that $f(n)$ is at least as large as the depth of the ``most interior'' point of a set (see there for a formal definition of the depth of a point).
In \cref{sec:colored} we give lower bounds for $g(n,c)$, the colored setting.
We finalize with a short discussion in \cref{sec:discussion}.

\section{Algorithmic tools}
\label{sec:algo_tools}

Our algorithms are based on subroutines using classical techniques. 
We quickly explain these subroutines here.

We will employ algorithms for the 
$k$-selection problem: given $n$ numbers and a value $k\le n$, 
compute the element that would be in the $k$th position, 
if the numbers would be sorted non-decreasingly.
It is well known that the $k$-selection problem can be solved performing 
a linear number of steps and comparisons between the input numbers;
input numbers are only compared, and no arithmetic operations with them are performed.
See Blum et al.~\cite{BlumFPRT73} or the textbook~\cite[Section 9.3]{CLRS_book}
for a description of the algorithm. Randomized variants are 
simpler~\cite[Section 9.2]{CLRS_book}.

We also use that a linear program with $2$ variables and $n$ constraints
can be solved optimally in $O(n)$ time; see Megiddo~\cite{Megiddo83a} for
a deterministic algorithm or the textbook~\cite[Chapter 4]{BergCKO08} for 
a simpler, randomized algorithm. Our use of linear programming is encoded 
in the following result. We use $CH(P)$ to denote the convex hull of $P$.

\begin{lemma}\label{lem:CH}
	Given a set $P$ of points in the plane and a ray $\rho$ 
	that intersects $CH(P)$, we can find in linear time the last intersection
	of $\rho$ with the boundary of $CH(P)$.
\end{lemma}
\begin{proof}
	\begin{figure}[htb]\centering
		\includegraphics[page=7,scale=1]{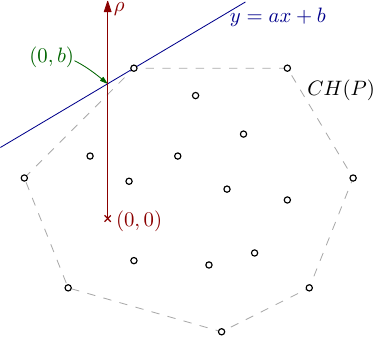}
		\caption{Proof of \cref{lem:CH}. The pair $(a,b)$ defining the blue
			line is a feasible solution to the linear program.}
		\label{fig:CH}
	\end{figure}
	Making a rigid motion, if needed, we may assume that $\rho$ is an upward 
	vertical ray that starts at the origin. 
	If all the points of $CH(P)$ are on the same side of the $y$-axis, 
	then $\rho$ is tangent to $CH(P)$ 
	and a point of $P$ is the last intersection of $\rho$ with the boundary 
	of $CH(P)$. We can test in linear time whether we are in this scenario,
	and select the last point of $P$ contained in $\rho$.
	
	It remains the case when there are points of $P$ on both sides of the $y$-axis.
	We then search for the line $\ell$ with equation $y=ax+b$ with minimum value $b$ such that
	all the points $p=(p_x,p_y)$ of $P$ lie below or on $\ell$. This is the following 
	LP with real variables $a,b$: 
	\[
		b^*=\min \{ b \mid \forall p\in P: ap_x+b \ge p_y \}.
	\]
	The point $(0,b^*)$ is the last intersection of the ray with $CH(P)$.
	See \cref{fig:CH} for a schema.
	If $(0,b^*)$ is not a vertex of $CH(P)$, then the line $y=a^*x+b^*$ 
	defined by an optimal solution $(a^*,b^*)$
	supports an edge of $CH(P)$.
	Since this is a linear program with $2$ variables and $n$ constraints,
	it can be solved in $O(n)$ time.
\end{proof}

Recall that a maximal matching is a matching where we cannot add any
additional edge and keep having a matching. In other words,
each edge of the graph has at least one vertex in common with
some edge of the matching.

\begin{lemma}\label{lem:maximal}
	Let $\sigma$ be a segment and let $\ell$ be its supporting line.
	Let $A$ be a set of at most $n$ points to one side of $\ell$
	and let $B$ be a set of at most $n$ points to the other 
	side of $\ell$.
	In $O(n\log n)$ time we can compute a maximal matching 
	in the bipartite graph
	\[
		G(A,B,\sigma) ~=~ (A\cup B, \{ ab\mid a\in A,~ b\in B, 
			\text{ $ab$ intersects $\sigma$}\}).
	\]
\end{lemma}
\begin{proof}
	Making a geometric transformation, we may assume that 
	$\sigma$ and $\ell$ are vertical, that $A$ is to the left of $\ell$,
	and $B$ to the right.
	Let $u$ and $v$ be the endpoints of $\sigma$ with $v$ above $u$.

	We define the function 
	$(\varphi_1,\varphi_2)=\varphi\colon\RR^2\setminus\ell \rightarrow \RR^2$
	by
	\begin{align*}
		\varphi_1(p) ~&=~ \text{slope of the line supporting $pu$},\\ 
		\varphi_2(p) ~&=~ \text{slope of the line supporting $pv$}.
	\end{align*}
	For points $a$ to the left of $\ell$ we have 
	$\varphi_1(a) < \varphi_2(a)$, while for points $b$
	to the right of $\ell$ we have $\varphi_1(b) > \varphi_2(b)$.
	Therefore, each number in the interval $[\varphi_1(a),\varphi_2(a)]$
	corresponds to a slope such that the line through $a$ with that 
	slope intersects $\sigma=uv$. A similar statement
	holds for $b$.

	\begin{figure}[htb]\centering
		\includegraphics[page=12,scale=1]{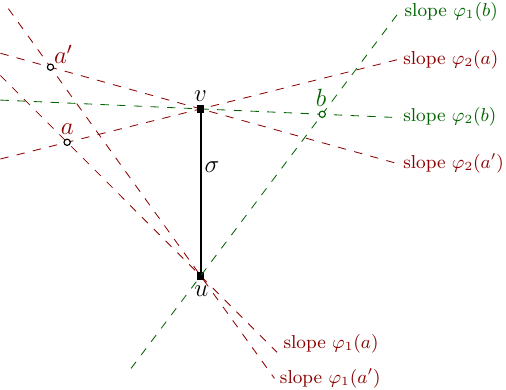}
		\caption{Proof of \cref{lem:maximal}. 
			The definition of the transformation
			$\varphi=(\varphi_1,\varphi_2)$.}
		\label{fig:varphi}
	\end{figure}
	
	Using that $a$ is to the left of $\ell$ and $b$ to the right,
	we note that $ab$ intersects $uv$ if and only if 
	$\varphi_1(a)\le \varphi_1(b)$ 
	and $\varphi_2(a)\ge \varphi_2(b)$.
	See \cref{fig:varphi}.
	One way to see this is noting that $ab\cap uv\neq \emptyset$
	if and only if the line supporting $au$ can be rotated counterclockwise 
	around $u$ until it becomes the line supporting $ub$, and 
	the line supporting $av$ can be rotated clockwise around $v$ 
	to turn it into the line supporting $bv$.
	This mapping $\varphi$ and an application is discussed in 	
	Cabello and Milinkovi{\'c}~\cite[Lemma 3]{CabelloM18},
	using point-line duality as an intermediary step in the discussion.

        We build the matching incrementally as follows. We process the points of
        $B$ one by one. For each point $b\in B$, we check whether it has an
        unmatched neighbor in $A$. If so, we match $b$ with $a$. It is clear
        that the matching produced by this procedure is maximal.

        In order to carry out the procedure efficiently, we process the points
        of $A$ and $B$ in order of increasing value of $\varphi_1$. At any point
        during the algorithm, we have processed all points $b\in B$ with
        $\varphi_1(b)\le t$ for some threshold $t$. We maintain the set $A'$
        of unmatched points $a\in A$ with $\varphi_1(a)\le t$. As we advance $t$,
        we may hit a point $a\in A$ with $\varphi_1(a)= t$ or a point $b\in B$ with
        $\varphi_1(b)= t$. In the first case, we simply add $a$ to $A'$.
        In the second case, any unmatched neighbor
        of $b$ is contained in $A'$, as all points $a\in A \setminus A'$ satisfy
        $\varphi_1(a)>\varphi_1(b)$. It remains to check whether $A'$ contains a point
        $a'$ with  $\varphi_2(a')\ge \varphi_2(b)$. If such a point exists, we include
        $a'b$ in the matching and remove $a'$ from $A'$.

        If we maintain $A'$ in a balanced binary search tree, ordered by the key
        $\varphi_2$, then insertions, deletions and the search for an element with
        large enough key can be performed in $ O(\log n)$ time.
        In total, the running time is $O(n \log n)$.
       \end{proof}

      Note, this algoritm gives a maximum matching with a suitable choice of  $a'\in A'$.
      If we  choose the vertex with the smallest
      $\varphi_2(a')$ value among the vertices $\varphi_2(a')\ge \varphi_2(b)$ in each step, we get a maximum matching.

\section{Balanced separation with a short path}
\label{sec:separation}
	
In this section we provide a structural result about splitting the convex hull of a point set with a single edge or with a 2-edge path in such a way that both sides contain a large fraction of the point set. This will be used later in our proofs of lower bounds. A very similar result can be found in \'Abrego and Fern\'andez-Merchant~\cite[Lemma~2]{AbregoF17}. We include a proof because their bound has a small error\footnote{Lemma 2 in~\cite{AbregoF17} is not correct for $n=4$ because a ceiling was missing in the bound. The authors have an updated, corrected version in arXiv.}, our approach is slightly different in the treatment of the triangular case (\cref{thm:split-triangle}), we develop a colored version (\cref{lem:separatingpath_colors1} and \cref{lem:separatingpath_colors2}), and because we discuss the algorithmic counterpart, a part that is not considered in~\cite{AbregoF17} and that forces us to rework a proof.
	
We first consider the case when the convex hull is a triangle and the partition can be with different numbers of points. This will be a tool for the general case. See \cref{fig:split-triangle1} to visualize the following statement.

\begin{figure}[htb]\centering
	\includegraphics[page=1,scale=1]{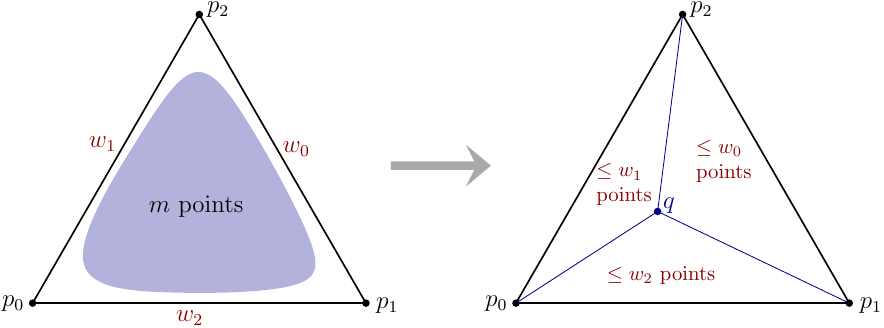}
	\caption{Statement in \cref{thm:split-triangle}.}
	\label{fig:split-triangle1}
\end{figure}

\begin{theorem}\label{thm:split-triangle}
	Assume that we have a triangle with vertices $p_0$, $p_1$, and $p_2$
	and in its interior there is a set $P$ of $m\ge 1$ points 
	such that $P \cup \{p_0,p_1,p_2\}$ is in general position. 
	For any integer weights $w_0,w_1,w_2$ such that $0\le w_0,w_1,w_2 < m$ 
	and $\ell:=w_0+w_1+w_2 > 2m -3$, there exist at least $\ell-2m+3>0$ points $q\in P$ 
	such that, for each $i\in \{0,1,2\}$,
	the triangle $\triangle(p_i q p_{i+1})$ contains at most $w_{i+2}$ 
	points of $P$ in its interior, where all indices are modulo $3$. %

	We can find $\ell-2m+3$ points with this property in linear time.
\end{theorem}
\begin{proof}
	In this proof, all indices are modulo $3$.
	For $i\in \{0,1,2\}$, consider a ray $r_i$ that starts at $p_{i-1}$ and 
	goes through $p_i$. We rotate $r_i$ around $p_{i-1}$ in the direction 
	towards $p_{i+1}$ until we pass $r_i$ over $m-w_i-1$ points of~$P$. 
	See Figure~\ref{fig:split-triangle2}, left, to visualize the case $i=1$.
	For any of the points $q\in P$ we did not scan over, the triangle
	$\triangle(p_{i-1} q p_{i+1})$ contains at most $w_i$ points of $P$
	in its interior; note that $q$ is not in the interior of
	$\triangle(p_{i-1} q p_{i+1})$.

	\begin{figure}[htb]\centering
		\includegraphics[page=2,scale=1]{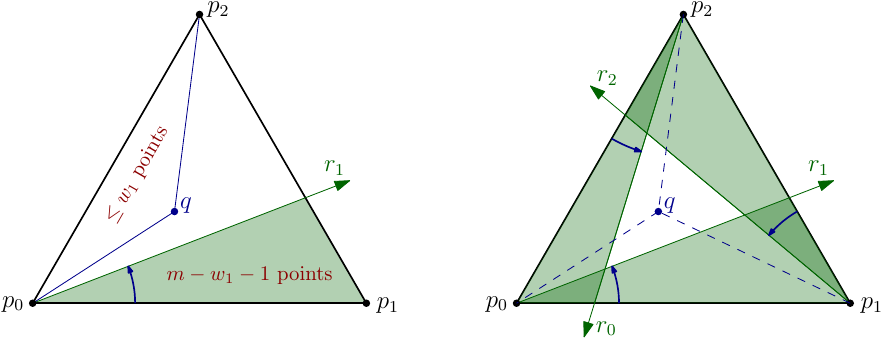}
		\caption{Left: rotating $r_1$ until we pass over $m-w_1-1$ points. 
			Any point not scanned by $r_1$ defines with $p_0$ and $p_2$ a
			triangle with at most $w_1$ points.
			Right: the part of the triangle that is not shadowed contains 
			at least $\ell-2m+3$ points.}
		\label{fig:split-triangle2}
	\end{figure}

	Some points of $P$ may be scanned more than once, but in total 
	we scan at most $3m-w_1-w_2-w_3-3 = 3m-\ell-3$ points. 
	So there are at least $m-(3m-\ell-3) = \ell-2m+3> 0$ points remaining,
	and each of them satisfies the desired property. See Figure~\ref{fig:split-triangle2}, right.
	
	To show the \emph{algorithmic claim}, we note that, for each $i\in \{0,1,2\}$,
	the points scanned by the rotation of $r_i$ can be computed in linear time.
	To see this, we associate the number 
	$\alpha(q)=\measuredangle(p_ip_{i-1}q)$ to each point $q\in P$. 
	We then compute the point $q_i$ of $P$ solving the $(m-w_i-1)$-selection
	problem with respect to $\{ \alpha(q)\mid q\in P\}$, which takes linear time.
	All points $q\in Q$ with $\alpha(q)\le \alpha(q_i)$ are marked as scanned.
	Note that we do not need to compute actual angles and that it suffices 
	to use orientation tests to compare angles.
	After performing this for $i=0,1,2$, the points that remain
	unmarked have the desired property.
\end{proof}		
	
As a special case we state the following corollary, which might be of its own interest.
	
\begin{corollary}\label{cor:triangle}
	Let $\Delta$ be a triangle with a set $P$ of $m\ge 1$ points in its interior. 
	Then there is a point of $P$ that splits $\Delta$ into three triangles, 
	such that none of these triangles contains more than $\lceil(2m-2)/3\rceil$ points of $P$ 
	in its interior.
\end{corollary}
\begin{proof}
	Perturb the points to general position, if needed, without changing the
	positive or negative orientation of any triple of points.
	Use now \cref{thm:split-triangle} with 
	$w_0=w_1=w_2 = \lceil(2m-2)/3\rceil$,
	which satisfy $w_0+w_1+w_2= 3 \lceil(2m-2)/3\rceil > 2m-3$,
	to obtain a point $q$.
	Each triangle defined by $q$ and any two vertices of $\Delta$
	contains at most $\lceil(2m-2)/3\rceil$ of $P$ in its interior. 
	Undoing the perturbation, some 
	points may move to the boundary of some of the triangles,
	but the number of points in the interior of a triangle cannot increase.
\end{proof}	
	
This result resembles the classical Centerpoint Theorem~\cite[Section 1.4]{Matousek2002},
which tells that for each set $P$ of $n$ points in the plane there exists a so-called 
centerpoint $q$ with the property that each open halfplane that does not contain $q$ 
has at most $2n/3$ of the points of $P$ inside.
However, the centerpoint does not need to be a point of $P$, it exists independently
of the shape of the convex hull,
and for some point sets it cannot be an element of $P$.
	
Recall that $CH(P)$ denotes the convex hull of $P$. A point $p \in P$ is \emph{extremal} for
(or an extreme point of) $P$ if it lies on the boundary of $CH(P)$. 
A \emph{$k$-separating path} for $P$ is a path $\pi$ spanned 
by vertices of $P$ and connecting two different extremal points of $P$ 
such that $CH(P)\setminus \pi$ has two parts, each containing at least $k$ points;
note that the points on the path are counted in no part.
The separation gets more balanced as $k$ increases.
See \cref{fig:separating1}.
The \emph{length} of such a path is its number of edges.	

\begin{figure}[htb]\centering
	\includegraphics[page=3,scale=1]{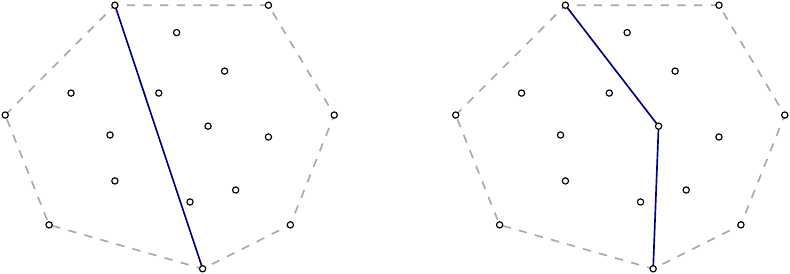}
	\caption{Left: $5$-separating path of length $1$.
			Right: $7$-separating path of length $2$.}
	\label{fig:separating1}
\end{figure}	
	
\begin{theorem}\label{thm:separatingpath}
	Let $P$ be a set of $n \geq 2$ points in general position in the plane. 
	Then there exists a $\lceil\frac{n-4}{3}\rceil$-separating path of length $1$ or $2$. %
	Such a separating path can be found in time linear in $n$.
\end{theorem}
\begin{proof}
    For $n \leq 4$ the statement is trivially true. So for the reminder of the proof assume that $n \geq5$.
	Let us set $r = \lceil (2(n-3)-2)/3\rceil =\lceil (2n-8)/3\rceil$.
	The intuition is that $r$ is the bound of \cref{cor:triangle}
	for $n-3 \geq 1$ points; in our current setting, $n$ is also counting the vertices
	of the triangle.
	We also set $k= \lceil (n-4)/3 \rceil \geq 1$ as $n \geq 5$.
	Note that $n-4 \le r+k \le n-3 $. 
	The task is to show the existence of a $k$-separating path of length $1$ or $2$.
	
	Choose an extremal point $q_0 \in P$ with the smallest $y$-coordinate.
	Let $q_1,\dots , q_{n-1}$ be the points $P\setminus \{ q_0\}$ sorted
	increasingly by the angle $\overline{q_0 q_i}$ makes with 
	the horizontal rightward ray from $q_0$.
	See \cref{fig:separating2}, left.

	\begin{figure}[htb]\centering
		\includegraphics[page=4,scale=1]{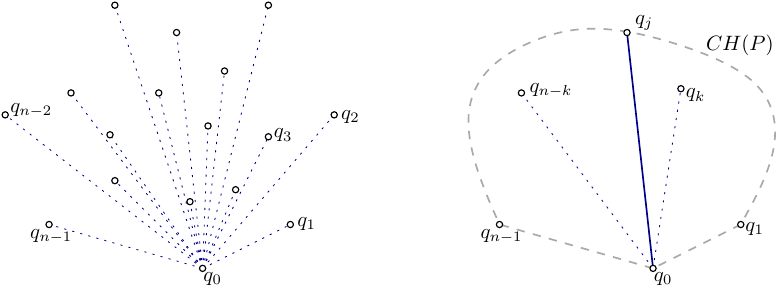}
		\caption{Proof of \cref{thm:separatingpath}.}
		\label{fig:separating2}
	\end{figure}	
	
	If between $q_k$ and $q_{n-k}$ there is some extremal
	point $q_j$ for $P$, which implies that $k<j<n-k$, 
	then the segment $q_0q_j$ is a $k$-separating path of length $1$
	and we are done. See \cref{fig:separating2}, right.
	Otherwise, the rays $q_0q_k$ and $q_0q_{n-k}$ intersect
	the same edge $e$ of $CH(P)$. Let $q_a q_b$ be the edge $e$,
	with $a<b$. 
	This means $a\le k < n-k \le b$ and the triangle 
	$\triangle(q_0 q_a q_b)$ has exactly $b-a-1$ points in its interior.
	See \cref{fig:separating4}, left.	Note that we may
	have $a=k$ and $b=n-k$.
	We have $n-2 \ge b-a\ge n-2k$.

	\begin{figure}[htb]\centering
		\includegraphics[page=5,width=\textwidth]{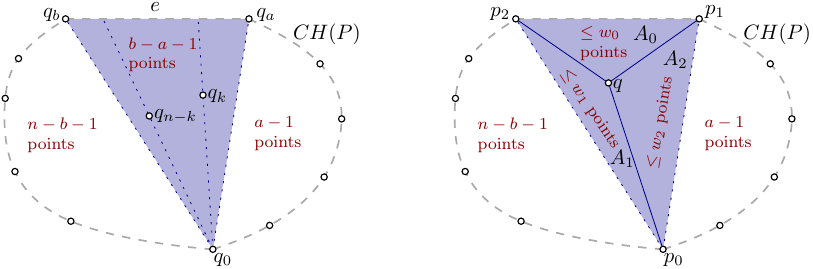}
		\caption{Continuation of the proof of \cref{thm:separatingpath}.}
		\label{fig:separating4}
	\end{figure}	
	
	We want to apply \cref{thm:split-triangle} to $\triangle(q_0q_aq_b)$
	and the $m=b-a-1 \geq n-2k-1 = n-2\lceil (n-4)/3 \rceil - 1 \geq n-2(n-2)/3-1 = (n+1)/3 \geq 1$ points of $P$ in its interior.
	To this end, set $p_0=q_0$, $p_1=q_a$, $p_2=q_b$,
	$w_0=r$, $w_1=r-(n-b-1)$, and $w_2=r-(a-1)$.
	See \cref{fig:separating4}, right.
	To apply \cref{thm:split-triangle}, we must verify that
	$w_0+w_1+w_2> 2m-3$:
	\begin{align*}
		w_0+w_1+w_2 ~&=~ r + \big( r -(n-b-1)\big) + \big( r - (a-1)\big) \\
		&=~ 3r - n + b - a +2 \\
		&\ge~ 3r - n + b-a +2 + \big(b-a - n +2\big) &&\hfill\text{~using $n-2 \ge b-a$}\\
		&=~ 3 \lceil (2n-8)/3\rceil  - 2n + 2(b-a) + 4 &&\hfill\text{~using $r=\lceil (2n-8)/3\rceil$} \\
		&\ge~ \big( 2n-8\big) - 2n + 2 (m+1) +4 &&\hfill\text{~using $m=b-a-1$} \\
		&>~ 2m-3.
	\end{align*}
	\cref{thm:split-triangle} guarantees the existence of a point $q\in P$ in the
	interior of $\triangle(p_0p_1p_2)=\triangle(q_0q_aq_b)$ that splits it into 
	three triangular pieces such that the interior of the triangle $\triangle(p_{i-1}qp_{i+1})$ 
	has at most $w_i$ points of $P$ (for $i=0,1,2$ and indices modulo $3$).
	
	We split $CH(P)$ into three parts $A_0,A_1,A_2$
	by removing the segments $qq_0=qp_0$, $qq_a=qp_1$, and $qq_b=qp_2$;
	the part $A_i$ is the one whose closure is disjoint from the relative 
	interior of $qp_i$ (for $i=0,1,2$).
	See \cref{fig:separating4}, right. 
	The points $q, q_0, q_a, q_b$ belong to no part, while all the other points of $P$
	belong to exactly one part.
	From the choices of weights $w_i$, each part contains at most $r$ points of $P$.
	For example, $A_1$ contains at most $(n-b-1)+w_1=r$ points.
	
	Let $B$ be a part $A_j$ that contains the most points of 
	$P\setminus \{ q,p_0,p_1,p_2 \}$. Let $\pi$ be the separating path of length $2$ that
	separates $A_j=B$ from $CH(P)\setminus A_j$; the path $\pi$ is the concatenation 
	of $p_{j-1}q$ and $qp_{j+1}$ (indices modulo $3$). 
	Let $B'$ be the other part of $CH(P)\setminus \pi$; it contains 
	$A_{j-1}$, $A_{j+1}$ and its common boundary (indices modulo $3$).

	By the pigeonhole principle, $B$ contains at least $\lceil (n-4)/3\rceil = k$ points.
	On the other hand, $B$ contains at most $r$ points, which means that 
	$B'$ contains $(n-3)-r \ge k$ points. It follows that $\pi$ is a $k$-separating 
	path of length $2$. 
	
	It remains to show the \emph{algorithmic claim}.
	Finding $q_0$ takes linear time by scanning the point set.
	The points $q_k$ and $q_{n-k}$ can be found in linear time because it is the $k$-
	and $(n-k)$-selection problem of $P$ with respect to the angle $q_0q_i$
	makes with the horizontal rightward ray from $q_0$. One does not need to compute the angle
	and can just use orientation tests.
	Using \cref{lem:CH}, we find in linear time the last intersection of the rays
	$q_0q_k$ and $q_0q_{n-k}$ with the boundary of $CH(P)$. Let $\ell$ be
	the line supporting those two intersections.
	
	We test whether $\ell$ has all the points of $P$ on the same side.
	If the test fails, we find the extremal point $q_j\in P$ 
	swept by the line $\ell$ when we move it away from $q_0$.
	See \cref{fig:separating6}, left.
	The point $q_j$ is then a vertex of $CH(P)$ and it lies in
	the cone defined by $q_kq_0 q_{n-k}$. In this case we are in the 
	first scenario and $q_0q_j$ is the desired separator.

	\begin{figure}[htb]\centering
		\includegraphics[page=6,scale=1]{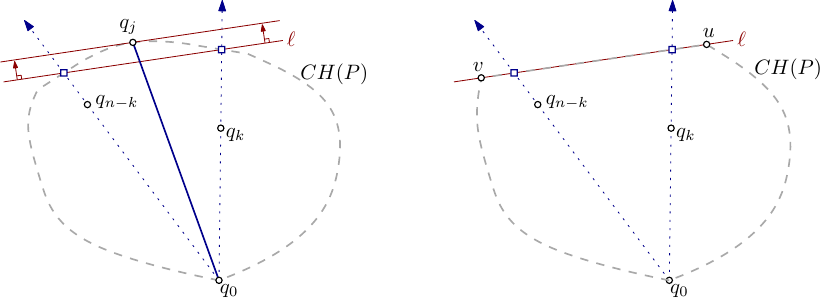}
		\caption{Algorithmic part of \cref{thm:separatingpath}.}
		\label{fig:separating6}
	\end{figure}	
	
	If the test is successful, that is, if all the points of $P$ lie on
	the same side of $\ell$, then $\ell$ supports an edge of $CH(P)$.
	This edge of $CH(P)$ is defined by the two points $u,v$ of $P$ that lie on 
	the line $\ell$, and they can be found in linear time.
	We then count how many points are in each region, that is,
	we figure out the indices $a$ and $b$ such that $u=q_a$ and $v=q_b$.
	We also compute the points in the interior of the triangle $q_0q_aq_b$,
	and use \cref{thm:split-triangle} to find in linear time
	the point $q$ used to split $\triangle(q_0q_aq_b)$.
	The rest of the proof is constructive and can be done in linear time
	by scanning and counting points.
\end{proof}

\section{Upper bounds}
\label{sec:upper}
	
In the following we provide upper bounds on the maximal size of connected matchings that exist for any given set of $n$ points.

We start with an upper bound for the uncolored case.
Consider $n$ points split into three sets $A_0, A_1, A_2$ of size roughly $\frac{n}{3}$, 
where the sizes of any two sets differ by at most one, and each $A_i$ lies on its own slightly curved blade of a three-bladed windmill; 
see \cref{fig:upperbound}, left. We use indices modulo three in the discussion.
We can form such a configuration so that each line determined by two points of $A_i$ 
separates $A_{i+1}$ from $A_{i+2}$, and no segment connecting one point of $A_i$ with 
one point of $A_{i+1}$ crosses any segment connecting two points of $A_{i+1}$. Hence, 
the set of all segments is separated into three parts where each part consists of 
segments connecting two points of $A_i$ or one point of $A_i$ and one point of 
$A_{i+1}$, and segments from different parts do not cross. Clearly, 
the size of the largest matching spanning $A_i\cup A_{i+1}$, if their sizes
differ by at most one, is $\min\{|A_i|,|A_{i+1}|\}$, 
and the largest of those values over $i\in \{ 0,1,2\}$
gives the largest connected matching. To be careful with the modulus of $n$, 
we note that there is a connected matching of maximum size $\lceil \frac{n}{3}\rceil$ 
when at least two of the sets have that size; when only one set has that size,
the largest connected matching has size $\lfloor \frac{n}{3}\rfloor$.
Thus, for each $n$ we have constructed a point set where the maximum connected matching
has size $\lceil \frac{n-1}{3}\rceil$.

\begin{figure}[tbh]\centering
	\hfill
	\includegraphics[page=8,scale=1]{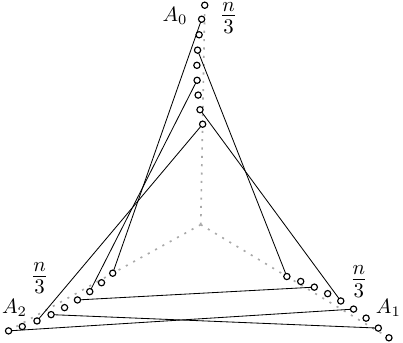} 
	~~~\hfill~~~ 
	\includegraphics[page=9,scale=1]{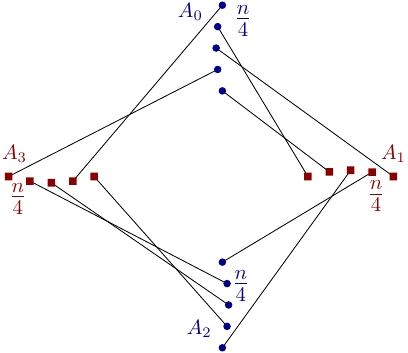}
	\hfill\mbox{}
	\caption{Upper bounds for (colored) connected matchings. Left: uncolored.
		Right: balanced $2$-colored.} 
	\label{fig:upperbound}
\end{figure}

Now, we provide an upper bound for the size of a connected matching in 
the balanced $2$-colored case. 
We consider a similar configuration. Recall
that the coloring is balanced: the cardinalities of each color class differ at most by one.
We split the points into four sets $A_0, A_1, A_2, A_3$ of size roughly $\frac{n}{4}$ 
so that the sizes of any two sets differ by at most one, and each $A_i$ lies on its own slightly curved blade of a four-bladed windmill.
The sets $A_0$ and $A_2$ contain only blue points, while $A_1$ and $A_4$ only red ones. 
See \cref{fig:upperbound}, right. In this configuration, bichromatic segments 
connecting points of $A_i$ with points of $A_{i+1}$ do not cross any other segments
(indices modulo $4$), so the size of the largest connected matching is $\frac{n}{4}$.
A maximum connected matching of size $\lceil\frac{n}{4}\rceil$ is attainable
when two sets of different colors have that cardinality, that is, when 
$n\equiv 2,3 \pmod{4}$. Thus, for each $n$ we have constructed a $2$-colored
point set where the maximum connected matching has size $\lceil \frac{n-1}{4}\rceil$.

\section{Lower bound for uncolored sets}
\label{sec:mono}	

We first consider the following special setting, 
depicted in \cref{fig:segment}, left.

\begin{figure}[tbh]\centering
	\includegraphics[page=10,width=.95\textwidth]{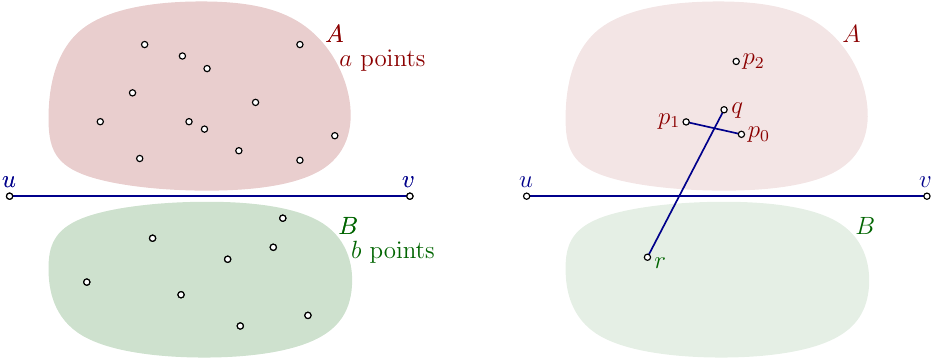}
	\caption{Left: Situation in \cref{lem:segment}.
		Right: edges added to the matching when $A$ has four points
		not in convex position.}
	\label{fig:segment}
\end{figure}
	
\begin{lemma}\label{lem:segment}
	Assume that we have a horizontal segment $uv$,
	a set $A$ of $a$ points above the line supporting $uv$,
	and a set $B$ of $b\le a$ points below the line supporting $uv$
	such that, for all $(p,q)\in A\times B$, the segment $pq$
	intersects $uv$, and $A\cup B\cup \{u,v \}$ consists
	of $a+b+2$ points in general position.
	Then, $A\cup B\cup \{u,v \}$ has a connected matching
	of size at least
	\[
	            m(a,b):=
		\begin{cases}
			1+b & \text{if $b\le a\le2b+3$},\\
			(a+3b+2)/5, & \text{if $2b+3 \le a \le 7b+3$},\\
			1+2b, & \text{if $a \ge 7b+3$}.
		\end{cases}~~~~~~~~~~~~~~~~~~~~~~~~~~~~~~~~~~~~~
	\] %
	Such a connected matching can be computed in $O(1+a \log a)$ time.
\end{lemma}
\begin{proof}	
	We first make two easy observations that will come in handy to follow
	the discussion:
	\begin{itemize}
		\item[(a)] A matching of $B$ onto $A$ with $b$ edges together
			with the edge $uv$ to ``connect'' them is
			a connected matching of size $b+1$. We want to improve upon
			this when the sides are unbalanced, in particular when $a$
			is larger than $2b\pm O(1)$.
		\item[(b)] If $A$ has a large subset $A'$ in convex position,
			then we can get a connected matching of size
			$\lfloor\frac{|A'|}{2}\rfloor$, for example
			by connecting ``antipodal'' points along the boundary of $CH(A')$.
	\end{itemize}
	
	We construct a connected matching $M$ iteratively as follows.
	At the start we add $uv$ to~$M$.
	While $|A|>|B|>0$ and $A$ has four points
	$p_0,p_1,p_2,q$ such that $q$ is in the interior of 
	$\triangle(p_0p_1p_2)$,
	we take an arbitrary point $r\in B$, add the edge $qr$ to $M$,
	and add to $M$ the edge $pp'$ of $\triangle(p_0p_1p_2)$ crossed by $qr$.
	See \cref{fig:segment}, right.
	Note that $\{ pp',uv,qr\}$ is a connected matching.
	Then we remove $p,p',q$ from $A$, and $r$ from $B$.
	With each repetition of this operation, we increase the size
	of the matching by two, remove three points from $A$, and remove
	a point from $B$. We repeat this operation until
	$B$ is empty, $|A|\le |B|$, or $A$ is in convex position, whatever
	happens first. 
	Let $k$ be the number of repetitions of this operation,
	let $A'$ and $B'$ be the subsets of $A$ and $B$, respectively,
	that remain at the end.
	Therefore, $M$ currently is a connected matching with $1+2k$ edges,
	$A'$ has $a-3k$ points, and $B'$ has $b-k$ points.
	
	We now consider the different conditions that hold at the end:
	\begin{description}
	\item[Condition (1)] If we finish because $B'$ is empty, then $k=b$ and 
		the matching $M$ has $1+2b$ edges. This scenario can happen
		only when $a\ge 3b$, because otherwise we run out
		of points of $A$ earlier.
	\item[Condition (2)] If we finish because $|A'|\le |B'|$, 
		we match the remaining points of $A'$ to $B'$ arbitrarily 
		and add those $|A'|$ edges to $M$; 
		since they cross $uv$, $M$ keeps being a connected matching.
		Because the cardinality of $A$ decreases at steps of size $3$
		and the cardinality of $B$ decreases at steps of size $1$,
		this means that $|A'|\le |B'|\le |A'|+1$,
		which implies that $a-3k\le b-k\le a-3k +1$, or equivalently,
		we have $a-2k\le b \le a-2k+1$.
		From this, because $k$ is an integer, we obtain that 
		$k=\lceil (a-b)/2\rceil$.
		The size of the connected matching $M$ is now 
		$1+2k+(a-3k) = 1+a - k = 1+ a- \lceil (a-b)/2\rceil =1+ \lfloor (a+b)/2\rfloor$.
		This scenario can happen for any $3b\ge a \ge b$. 
	\item[Condition (3)] If we finish because $A'$ does not have any $4$ points 
		with the desired condition, the key observation is to note that
		$A'$ is in convex position. (This is also true if $|A'|\le 3$.)
		We consider two connected matchings and take the best of both. 
		
		The first matching is obtained
		by adding to $M$ a matching between all the vertices of $B'$
		and any subset of $A'$ with $|B'|$ points. 
		The second matching, which we denote by $M'$, is obtained
		by taking a connected matching of the points $A'$, that
		is in convex position. Note that this is a matching \emph{within} $A$.
		We take the larger matching of $M$ and $M'$.
		
		The connected matching $M$ has size $1+2k+ (b-k) =1+b+k$.
		The other connected matching, $M'$, has $\lfloor\frac{|A'|}{2}\rfloor = 
		\lfloor\frac{a-3k}{2}\rfloor \ge \frac{a-3k-1}{2}$ edges. 
		Therefore, in this outcome we get a connected matching of size 
		\[
			\max\left\{ 1+b+k, \frac{a-3k-1}{2}  \right\}.
		\]
		The first term increases with $k$, the second term decreases with $k$,
		and the two terms are equal when $k$ takes the value $k_0:=(a-2b-3)/5$.
		At $k=k_0$ the expression takes the value $(a+3b+2)/5$.
		However, because in each step we remove $3$ points from $A$ and $1$ point
                from $B$, we have some additional constraints, as follows.
		\begin{align*}
			k\le a/3 &:~~~~~~ k_0\le a/3 ~~~\Longleftrightarrow~~~ 2a\ge -6b-9, \text{ always true}.\\
			k\le b&:~~~~~~ k_0\le b ~~~\Longleftrightarrow~~~ a\le 7b+3.\\
			k\ge 0&:~~~~~~ k_0\ge 0 ~~~\Longleftrightarrow~~~ a\ge 2b+3.
		\end{align*}
		Therefore, if $a<2b+3$, then $k_0<0$ and 
		the maximum is always attained at the function $1+b+k$, which in
		the worst case takes value $1+b$.
		If $a > 7b+3$, then $k_0>b$ and for all valid values of $k$ ($k\le b$)
		the maximum is given by $\frac{a-3k-1}{2}$; its minimum value is
		at $k=b$, giving $\frac{a-3b-1}{2}$.
		Summarizing this outcome, we get a connected matching whose size is
		bounded from below by the following function
		\[
		            \tilde{m}(a,b):=
			\begin{cases}
				1+b & \text{if $b\le a\le 2b+3$},\\
				(a+3b+2)/5, & \text{if $2b+3 \le a \le 7b+3$},\\
				(a-3b-1)/2, & \text{if $a \ge 7b+3$}.
			\end{cases}
		\]
		Note that this function is ``continuous'' at the boundary cases,
		which is a ``good indication''.
	\end{description}
	
	Since we have given a construction that can finish with $3$ 
	different conditions, we have to consider the worst case among 
	those scenarios, and show that in each case
	$m(a,b)$ is a lower bound on the size of the connected matching.
		
	We first compare the outcomes under Conditions (1) and (3) and
	see that actually $m(a,b)$ describes the worst case among them.
	\begin{align*}
		\text{If $b\le a<2b+3$}:&~~~~~~ 1+b \le 1+2b \text{~~~always}\\
		\text{If $2b+3 \le a \le 7b+3$}:&~~~~~~ \frac{a+3b+2}{5} \le 1+2b
					~~~\Longleftrightarrow~~~ a \le 7b+3\\
		\text{If $a \ge 7b+3$}:&~~~~~~ \frac{a-3b-1}{2} \ge 1+2b
					~~~\Longleftrightarrow~~~ a \ge 7b+3.
	\end{align*}
        It remains to compare the outcome under Condition (2) and $m(a,b)$;
	we will see that the worst case is never in this outcome.
	In some cases we compare against $(a+b+1)/2 \le 1+\lfloor (a+b)/2 \rfloor$,
	as it suffices and it is easier to manipulate.
	\begin{align*}
		\text{If $b\le a<2b+3$}:&~~~~~~ 1+b\le 1+\left\lfloor \frac{a+b}{2} \right\rfloor 
			~~~\Longleftrightarrow~~~ a\ge b \text{~~~always}\\
		\text{If $2b+3 \le a \le 7b+3$}:&~~~~~~ \frac{a+3b+2}{5} \le \frac{a+b+1}{2}
					~~~\Longleftrightarrow~~~ b\le 3a+1 \text{~~~always}\\
		\text{If $a \ge 7b+3$}:&~~~~~~ 1+2b \le \frac{a+b+1}{2}
					~~~\Longleftrightarrow~~~ a \ge 3b+1.
	\end{align*}
	We conclude that $m(a,b)$   indeed gives a lower bound on the size
	of a connected maximum matching.
	
	\begin{figure}[htb]\centering
		\includegraphics[page=11,width=\textwidth]{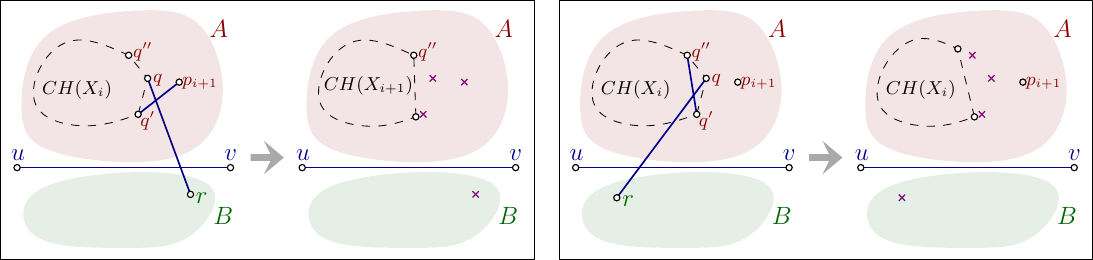}
		\caption{Algorithmic part of \cref{lem:segment}. Two cases
			may arise when inserting $p_{i+1}$: we may match $p_{i+1}$
			to $B$ (left) or we may match the other three points
			that together with $p_{i+1}$ form a non-convex $4$-tuple.
			Crosses denote points that are deleted.}
		\label{fig:segment2}
	\end{figure}

	It remains to discuss the \emph{algorithmic claim}.
	We only discuss the case of $a\ge 1$. Since $a\ge b$, we have $O(a)$ points
	in total. The proof is constructive and most of it is just simple book keeping
	of the sizes of the sets. The only complicated aspect of the algorithm
	is finding the $4$ points of $A$ that are not in convex position,
	or recognize that $A$ is in convex position.
	For this we employ an incremental algorithm to compute $CH(A)$
	by adding the points by increasing $x$-coordinate.
	Let $p_1,\dots,p_a$ be the points of $A$ sorted by increasing $x$-coordinate.
	For each index $i$, let $A_i$ be the prefix $\{p_1,\dots, p_i\}$.

	We maintain a connected matching $M$ and a subset $X_i\subseteq A_i$ 
	such that: (i) $X_i$ is in convex position, and
	(ii) $A_i\setminus X_i$ are endpoints of the connected matching we maintain. 
	This means that $X_i\cup \{ p_{i+1},\dots, p_a\}$ is the set of points $A$
	maintained through the iterations of the constructive proof.
	For $X_i$ we maintain its convex hull, $CH(X_i)$, as a linked list with a finger 
	to its rightmost point; in general, $p_i$ is \emph{not} the rightmost point of $X_i$
	because it may have been matched.
	
	When we add the next point, $p_{i+1}$, we compute $CH(X_i\cup \{p_{i+1}\})$ from $CH(X_i)$.
	If in the process we do not delete any point of $X_i$, meaning that all points
	of $X_i\cup \{p_{i+1}\}$ are extremal, we just set $X_{i+1}=X_i\cup \{p_{i+1}\}$,
	move the finger to $p_{i+1}$ because it is the rightmost point of $X_{i+1}$, 
	and move to the next point, $p_{i+2}$.
	If in the process we delete some point $q$ or $X_i$, let $q'$ and $q''$ 
	the neighbors of $q$ along the boundary of $CH(X_i)$. 
	The triangle $\triangle(p_{i+1}q'q'')$ contains $q$ in its interior. 
	See \cref{fig:segment2}. 
	In this case we make once the operation described in the constructive proof:
	select any point $r$ from $B$, add to the matching $qr$ and 
	the edge $e$ of $\triangle(pq'q'')$ it crosses. 
	Now we have to remove the points $q$ and two other points of $e$. 
	For this we undo the changes we made to $CH(X_i)$, so that we get $CH(X_i)$ back.
	We remove the points of $\{ q,q',q'' \}\cap X_i$ that were matched, 
	which takes constant time;
	we may have to update the finger to the point with largest $x$-coordinate.
	If $p_{i+1}$ is to be removed because it was matched, 
	we have finished and move to the next point, $p_{i+2}$; this
	is the case in the left of \cref{fig:segment2}.
	Otherwise, we try to reinsert $p_{i+1}$ again, which may trigger another
	iteration adding another two edges to the matching;
	this is the case on the right of \cref{fig:segment2}.
	
	After sorting the points of $A$, we spend $O(1)$ time per point, if we do not
	add any edge to the matching, and $O(1)$ time per edge added to the matching.
	Therefore, in total we spend $O(a\log a)$ time.

\end{proof}

Note that the bound $m(a,b)$ of \cref{lem:segment}
is monotone increasing in $a$ and in $b$, also when we take $a$ and $b$
as real values (with $b\le a$ always.)
Moreover, when $a+b$ remains constant, then $m(a,b)$ is larger
for larger $b$. This means $m(a,b)\le m(a-1,b+1)$ whenever $b\le a-2$.

\begin{theorem}\label{thm:uncolored}
	Let $P$ be a set of $n\ge 2$ points in general position in the plane. 
	Then $P$ has a connected matching of size at least $(5n+1)/27$
	which can be computed in $O(n\log n)$ time.
\end{theorem}
\begin{proof}
	By \cref{thm:separatingpath} we know that there is 
	a $\lceil\frac{n-4}{3}\rceil$-separating path $\pi$ of length 1 or 2 for~$P$. 
	Let $A$ and $B$ be the sets of points of $P$ on each side of $\pi$, 
	such that $|A|\ge |B|$. Note that the vertices of $\pi$ 
	belong neither to the set $A$ nor to the set $B$, which means that $n-3\le |A|+|B| \le n-2$.
	Therefore we have
	\[
		\left\lceil\frac{n-4}{3}\right\rceil ~\le~ |B| ~\le~ |A| 
		~\le~ n-3-\left\lceil\frac{n-4}{3} \right\rceil 
		~=~ \left\lfloor\frac{2n-5}{3} \right\rfloor.
	\]
	Each edge connecting a point of $A$ to a point of $B$ crosses $\pi$.	

	If $\pi$ consists of a single edge $e$, then we match all points of 
	$B$ to points of $A$ arbitrarily, and include $e$ also 
	in the matching.
	Since all these edges intersect $e$, they form a connected matching
	of size $1+|B| \ge \lceil\frac{n-1}{3}\rceil \ge \tfrac{5n+1}{27}$. 
	(This last inequality holds for $n\ge 2$.)

	For the remainder of this proof we assume that $\pi$ has length two, 
	and denote its edges by $e_1$ and $e_2$. 
	We build a \emph{maximal} matching $M_1$ from 
	$B_1\subseteq B$ to $A_1\subseteq A$ with edges that cross~$e_1$. 
	This means that $|A_1|=|B_1|$ and there is no point in $A\setminus A_1$
	that can be connected to a point in $B\setminus B_1$ by crossing $e_1$.
	Set $A_2= A\setminus A_1$ and $B_2= B\setminus B_1$;
	Each segment connecting a point in $A_2$ to a point of $B_2$ 
	must cross $e_2$ because it does not cross $e_1$.
	We make an arbitrary matching $M_2$ connecting each point of $B_2$ 
	to points of $A_2$; this can be done because 
	$|B_2|=|B|-|M_1|\le |A|-|M_1|= |A_2|$.
	We add $e_1$ to $M_1$ and $e_2$ to $M_2$ so that 
	$M_1$ and $M_2$ become connected matchings with $|M_1|+|M_2| = 2+|B|$.
	
	If $M_1$ or $M_2$ has size at least $\tfrac{5n+1}{27}$, then we are done.
	Therefore, we can restrict our attention to the case
	when $|M_1|, |M_2| \le \tfrac{5n+1}{27}$.
	Since $|A_1|=|B_1|=|M_1|-1 \le \tfrac{5n-26}{27}$, we have	
	\begin{align*}
		|B_2|~&=~|B|-|B_1| ~\ge~ \left\lceil\frac{n-4}{3}\right\rceil - \frac{5n-26}{27}
			~\ge~  
			\frac{4n-10}{27}.
	\end{align*}
	
	We apply \cref{lem:segment} to the segment $e_2$ with $A_2$ and $B_2$
	to get a connected matching, where $a=|A_2|$ and $b=|B_2|$.
	Since the lower bound $m(a,b)$ of \cref{lem:segment}  
	is monotone increasing in $b$, even when $a+b$ is fixed, 
	we get a worst-case lower bound by evaluating it at 
	\begin{align*}
		b ~&:=~ \frac{4n-10}{27} \le |B_2| \\
		a ~&:=~ \frac{13n-19}{27} ~=~ (n-3) - 2\cdot \frac{5n-26}{27} - \frac{4n-10}{27} \\
			&\phantom{:=~ \frac{13n-19}{27}}~~\le~ (n-3) - |A_1|-|B_1|- b ~=~ |A_2| + |B_2| -b,
	\end{align*} 
	because $a+b\le|A_2|+|B_2|$. Note that for this choice of $a$ and $b$ 
	we indeed have $b\le a$ for $n\ge 2$.
	To evaluate the function $m(a,b)$ of \cref{lem:segment}, 
	the values $a=\frac{13n-19}{27}$ and $b=\frac{4n-10}{27}$ fall 
	in the regime $2b+3\le a \le 7b+3$, when $n\ge 16$, because
	\[
		2b+3 ~=~ \frac{8n+61}{27} ~~~~~\stackrel{\mathclap{16\le n}}{\le}~~~~~ \frac{13n-19}{27} ~=~ a 
			~\le~ \frac{28n+11}{27} ~=~ 7b+3.
	\]
	In this case, when $n\ge 16$, we obtain the worst-case lower bound
	\[
		\frac{a+3b+2}{5} ~=~ 
		\frac{1}{5} \cdot \left(\frac{13n-19}{27} + 3\cdot \frac{4n-10}{27} +2\right)
			~=~ \frac{1}{5} \cdot \frac{25n+5}{27} \ge \frac{5n+1}{27}.
	\]
	
	For $2\le n\le 15$, we have to evaluate the function $m(a,b)$ of 
	\cref{lem:segment} in the regime $b\le a\le 2b+3$,
	and the lower bound we obtain is
	\[
		1+b ~=~ 1+ \frac{4n-10}{27} ~=~ \frac{4n+17}{27} 
		~~~~~~\stackrel{\mathclap{16\ge n}}{\ge}~~~~~~ \frac{5n+1}{27}.
	\]
	This covers all options for $n$ and concludes the proof of the lower bound $(5n+1)/27$.
	
	It remains to discuss the \emph{algorithmic claim}.
	The computation of the separating path via \cref{thm:separatingpath}
	takes $O(n)$ and the computation of the maximal matching takes $O(n\log n)$
	using \cref{lem:maximal}. \cref{lem:segment} takes $O(n\log n)$ time.
	The remaining tasks are simple book keeping of cardinalities of sets.
\end{proof}

\subsection{Sets with deep points}
\label{sec:deep}
We define the \emph{depth} $d(p)$ of a point $p \in P$ as the minimum number of points 
that need to be removed from $P$ so that $p$ lies on the boundary of the 
convex hull of the remaining points. This implies that any line through 
$p$ has at least $d(p)$ points of $S$ on both of its sides.
If the set $P$ of points is in convex position, then $d(p)=0$ for all $p\in P$.
However, for some point sets in general position, we may have some point at depth $(n-2)/2$.
	
\begin{theorem}
\label{thm:deep}
	Let $P$ be a set of $n$ points in general position in the plane 
	and let $p$ be a point of $P$ with the largest depth $d(p)$ in $P$. 
	Then $P$ has a connected matching of size at least~$d(p)$.
\end{theorem}
\begin{proof}
	Let $\ell$ be a line that passes through $p$ and an arbitrary extremal 
	point $a$ of~$P$. Without loss of generality we may assume that $pa$
	is horizontal with $a$ to the right of $p$.
	The edge $e=pa$ is our first matching edge and we will 
	construct additionally $d(p)-1$ matching edges that all intersect $e$.
	See \cref{fig:deep}.

	\begin{figure}[htb]\centering
		\includegraphics[page=13,scale=1]{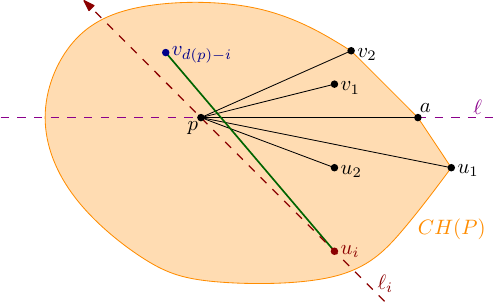}
		\caption{Proof of \cref{thm:deep}. The segment $u_iv_{d(p)-i}$ intersects
			$ap$ for each $i$ with $1\le i<d(p)$.}
		\label{fig:deep}
	\end{figure}

	Label the points strictly above $\ell$ as $v_1,\dots, v_k$
	in counterclockwise order around $p$, from $a$ onwards.
	Note that $k\ge d(p)$.
	Label the points strictly below $\ell$ as $u_1,\dots, u_{k'}$
	in clockwise order around $p$, from $a$ onwards. Note that $k'\ge d(p)$,
	Let $\ell_i$ be the oriented line passing through $u_i$ and then $p$.
	For each $i$ with $1\le i< d(p)$, the points $v_1,\dots, v_{d(p)-i}$
	are to the right of $\ell_i$. This is so because otherwise
	to the right of the $\ell_i$ we would have  
	$\{a,u_1,\dots, u_{i-1}\}$ and a subset of $\{v_1,\dots, v_{d(p)-i-1}\}$,
	which in total has $1+(i-1)+d(p)-i-1=d(p)-1$ points,
	contradicting the fact that $p$ is at depth $d(p)$.
	Moreover, because $v_{d(p)-i}$ is to the right of $\ell_i$ and $a$
	is an extreme point of $CH(P)$, the segment $u_iv_{d(p)-i}$ intersects
	the segment $e=pa$.
	It follows that the segments $u_1v_{d(p)-1}, u_2v_{d(p)-2}, \dots, u_{d(p)-1}v_{1}$
	together with $e=pa$ form a connected matching of size $d(p)$.
\end{proof}
	
Note that the bound is tight for four points, when one of them is in the interior of the convex hull.

\section{Lower bound for colored sets}
\label{sec:colored}

For this section, $P$ denotes a set of $n$ points in general position in the plane
with a balanced $c$-coloring. This means that each of the $c$ color classes
has roughly $n/c$ points.
To avoid carrying floors and ceilings, which make the computation more cumbersome,
in our results we will not optimize additive constants.

For colored sets we will prove our lower bounds using separating paths, as in the 
uncolored case.  The main difference is that we want that each edge of the separating
path connects points with different colors, as otherwise they can not be used as matching edges. 
For this, we say that a \emph{polychromatic $k$-separating path} is a $k$-separating path
where each edge of the path connects points with different colors.
To show the existence of polychromatic $k$-separating path, for a suitable $k$, 
we use \cref{thm:split-triangle} in such a way that there are enough
candidate points to split the triangle into the required weighted subtriangles.
A sufficiently large number of points allow us to have flexibility of choosing the 
color of the points in the separating path. 
 
We start by showing colored variants of \cref{thm:separatingpath}. We provide two
results, each of them better for a different range of $c$.

\begin{lemma}\label{lem:separatingpath_colors1}
	For $c\ge 4$ and sufficiently large $n$, 
	there exists a polychromatic $\left(\frac{(c-3)n}{3c}-3\right)$-separating path 
	for $P$ of length $1$ or $2$.
	Such a separating path can be found in time linear in~$n$.
\end{lemma}
\begin{proof}
    We closely follow the proof of \cref{thm:separatingpath}.
    As it was done there, we set $k= \lceil (n-4)/3 \rceil$, 
	which means that $\frac{n-4}{3} \le k \le \frac{n-2}{3}$,
	and define $q_0, q_1,\dots , q_{n-1}$ by sorting the points radially
	around the point $q_0$ with minimal $y$-coordinate.
	 
	Consider first the case where between the points $q_k$ and $q_{n-k}$ there is 
	an extremal point~$q_j$.
	If $q_0$ and $q_j$ have different colors, then we take $q_0q_j$ as the separating path,
	which is a $k$-separating path for $k\ge \frac{n-4}{3} \ge \frac{(c-3)n}{3c}-3$,
	whenever $c\ge 2$.
	Otherwise, we take a point $q_\ell$ with $\frac{n}{4}-1 \le \ell\le \frac{3n}{4}+1$
	such that $q_\ell$ has a color different than $q_0$ and $q_j$.
	See \cref{fig:separating14}, left.
	Such a point exists because there are at least $(\frac{3n}{4}+1)-(\frac{n}{4}-1) -1 = \frac{n}{2} +1$
	points, and no color class has more than $\frac{n}{2}+1$ points. 
	Note that $q_j$ is also from this interval of points, but as $q_0$ has the same color as $q_j$ 
	there is at least one point with a different color among the $\frac{n}{2} +1$ points.
	Then, the path $q_0q_\ell q_j$ is a $(\frac{n}{4}-2)$-separating path of length $2$:
	it has at least $\ell-1\ge \frac{n}{4}-2$ on one side, and at least $k\ge \frac{n-4}{3}$
	on the other side. Finally, we note that $\frac{n}{4}-2 \ge \frac{(c-3)n}{3c}-3$ for $c\ge 2$.
	(It may be that $q_\ell$ is an extreme point, and therefore $q_0q_\ell q_j$ splits
	$CH(P)$ into three parts, but then two of them are on the same side of the path.)
	
	\begin{figure}[htb]\centering
		\includegraphics[page=14,scale=1]{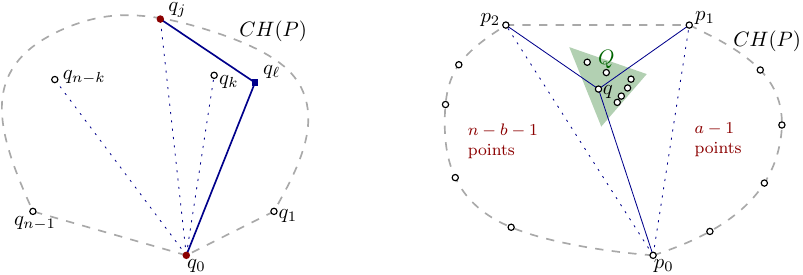}
		\caption{Proof of \cref{lem:separatingpath_colors1}. 
			Left: Schema where $\ell$ satisfies $\frac{n}{4}-1 \le \ell \le \frac{3n}{4}+1 $. 
			Right: All the points in $Q$ give a good enough partition, and $Q$
				has points with at least $4$ colors.}
		\label{fig:separating14}
	\end{figure}	

	Now we turn to the case where between $q_k$ and $q_{n-k}$ there is no extremal
	point. This means that there is an edge $q_aq_b$ of $CH(P)$ such that $q_k$ and $q_{n-k}$
	are in the triangle $\triangle(q_0 q_a q_b)$, which also
	means that $a\le k< n-k\le b$. 
	The triangle $\triangle(q_0 q_a q_b)$ has $m=b-a-1\le n-3$ points in the interior.
	We set $p_0=q_0$, $p_1=q_a$, $p_2=q_b$, 
	\[
		w_0= \left\lceil\left( \frac{1}{c}+\frac{2}{3}\right) n\right\rceil, 
		~~~ w_1= w_0 - (n-b-1), ~\text{ and }
		w_2= w_0 - (a-1).
	\]
	We first note that
	\begin{align*}
		w_0 ~&>~ 0,\\
		w_1 ~&=~ w_0 - (n-b-1) \ge \left( \frac{1}{c}+\frac{2}{3}\right)n - k +1
				> \frac{2n}{3} - \frac{n-2}{3} +1 > 0,\\
		w_2 ~&=~ w_0 - (a-1) \ge \left( \frac{1}{c}+\frac{2}{3}\right)n - k +1
				> 0,
	\end{align*}
	and then we note that
	\begin{align*}
		w_0 + w_1 + w_2 ~&\ge ~ 3 \left( \frac{1}{c}+\frac{2}{3}\right) n - (n-b-1) - (a -1)\\
			&=~ 2n + \frac{3n}{c}- n + b -a +2 \\
			&\ge~ \frac{3n}{c} + n + b-a +2\\
			&\ge~ \frac{3n}{c} + (m+3) + (m+1) +2, \ \text{~~using $m=b-a-1$ and $n-3 \ge m$}\\
			&\ge~ \frac{3n}{c} + 2m + 6.
	\end{align*}
	This means that, using \cref{thm:split-triangle} we get a set $Q\subset P$
	of at least $w_0 + w_1 + w_2 - 2m+3 \ge  \frac{3n}{c} + 9$ points 
	such that each $q\in Q$ satisfies the conclusion of \cref{thm:split-triangle}: 
	the interior of each triangle $\triangle(p_{i-1}qp_{i+1})$ has at most $w_i$ points of $P$ 
	(for $i=0,1,2$ and indices modulo $3$).
	See \cref{fig:separating14}, right.

	Since each color class has at most $\lceil \frac{n}{c} \rceil\le \frac{n}{c}+1$ points, 
	$Q$ has points with at least four different colors.
	Let $q$ be a point of $Q$ with a color different than $p_0,p_1,p_2$. 
	We can now use $qp_0,qp_1,qp_2$ to split $CH(P)$, as it was done in the proof
	of \cref{thm:separatingpath}, and to select the piece $B$ with the largest
	number of points. As it happened there, $B$ has at least $\frac{n-4}{3}$ points
	by the pigeonhole principle, and 
	it has at most $w_0\le \left( \frac{1}{c}+\frac{2}{3}\right) n$ points by construction,
	which means that the other side has at least
	\[
		n-3-w_0 ~\ge~ n-3 - \left( \frac{1}{c}+\frac{2}{3}\right) n = \frac{(c-3)n}{3c}-3
	\]
	points.
	
	The algorithm to compute the separating path is very similar to the algorithm in \cref{thm:separatingpath}.	
\end{proof}

\begin{theorem}\label{thm:separatingpath_colors1}
	Assume that $c\ge 4$ and $n$ is sufficiently large.
	Let $P$ be a set of $n$ points in general position in the plane
	with a balanced $c$-coloring.
	Then $P$ has a polychromatic connected matching of size at least 
	$\frac{(c-3)n}{6c}-\frac{1}{2}$
	which can be computed in $O(n)$ time.
\end{theorem}
\begin{proof}
	We use \cref{lem:separatingpath_colors1} to compute a polychromatic
	$\left(\frac{(c-3)n}{3c}-3\right)$-separating path $\pi$
	for $P$ of length $1$ or $2$.
	Let $A$ and $B$ be the sets on one side and the other side of $\pi$,
	and set $k=\min\{ |A|,|B|\} \ge \frac{(c-3)n}{3c}-3$. 
	We can compute a polychromatic matching $M$ of size $k$ between $A$ and $B$ greedily:
	at each step, we match a point from $A$ and a point of $B$ with different colors
	from the two most popular color classes; in this way, different color classes
	differ by at most one through the whole procedure, and all the points
	in the smallest side get matched.
	Since each edge of $M$ crosses $\pi$, at least one of the two (or fewer) edges of $\pi$, 
	say $e$, is intersected by $|M|/2$ edges of $M$. The edges of $M$ intersecting $e$
	together with $e$ form a polychromatic matching of size at least
	$1+\frac{k}{2}\ge 1+ \frac{(c-3)n}{6c}-\frac{3}{2} = \frac{(c-3)n}{6c}-\frac{1}{2}$.
	The computation in linear time is easy after obtaining the separating path 
	of \cref{lem:separatingpath_colors1}.
\end{proof}

\begin{lemma}\label{lem:separatingpath_colors2}
	For $c\ge 2$ and sufficiently large $n$, 
	there exists a polychromatic path $\pi$	with at most $3$ edges, 
	and two sets $P',P''\subset P$, each with at least $\frac{(c-1)n}{3c}-4$ points,
	such that each edge connecting a point from $P'$ to a point of $P''$ intersects $\pi$.
\end{lemma}	
\begin{proof}
	The path $\pi$ we are searching for is essentially a 
	polychromatic $k$-separating path of length at most $3$, 
	but now the path may self-intersect and the regions are not obvious.
	
	We closely follow the proof of \cref{lem:separatingpath_colors1}.
	The only difference is that we set
	\[
		w_0= \left\lceil\left( \frac{1}{3c}+\frac{2}{3}\right) n\right\rceil, 
		~~~ w_1= w_0 - (n-b-1), ~\text{ and }
		w_2= w_0 - (a-1).
	\]
	(In the proof of \cref{lem:separatingpath_colors1} we had $\frac{1}{c}$ instead of $\frac{1}{3c}$.)
	Like before, we note that
	\begin{align*}
		w_0 ~&>~ 0,\\
		w_1 ~&=~ w_0 - (n-b-1) \ge \left( \frac{1}{3c}+\frac{2}{3}\right)n - k +1
				> \frac{2n}{3} - \frac{n-2}{3} +1 > 0,\\
		w_2 ~&=~ w_0 - (a-1) \ge \left( \frac{1}{3c}+\frac{2}{3}\right)n - k +1
				> 0,
	\end{align*}
	and then we note that
	\begin{align*}
		w_0 + w_1 + w_2 ~&\ge ~ 3 \left( \frac{1}{3c}+\frac{2}{3}\right) n - (n-b-1) - (a -1)\\
			&=~ 2n + \frac{n}{c}- n + b -a +2 \\
			&\ge~ \frac{n}{c} + n + b-a +2\\
			&\ge~ \frac{n}{c} + (m+3) + (m+1) + 2, \ \text{~~using $m=b-a-1$ and $n-3 \ge m$}\\
			&\ge~ \frac{n}{c} + 2m + 6.
	\end{align*}
	This means that, using \cref{thm:split-triangle} we get a set $Q\subset P$
	of at least $w_0 + w_1 + w_2 - 2m+3 \ge  \frac{n}{c} + 9$ points 
	such that each $q\in Q$ satisfies the conclusion of \cref{thm:split-triangle}: 
	the interior of each triangle $\triangle(p_{i-1}qp_{i+1})$ has at most $w_i$ points of $P$ 
	(for $i=0,1,2$ and indices modulo $3$).
	Recall \cref{fig:separating14}, right.

	Since each color class has at most $\lceil \frac{n}{c} \rceil\le \frac{n}{c}+1$ points, 
	$Q$ has points with at least \emph{two different colors}.
	Let $q_1,q_2$ be points of $Q$ with different colors.
	If one of them has a color different than the three points $p_0,p_1,p_2$,
	we can continue as usual.
	Otherwise, we connect each point $p_i$ ($i=1,2,3$) with a point $q_j$ ($j=1$ or $j=2$)
	that has a different color. We also connect $q_1q_2$. See \cref{fig:separating15}.
	These four edges define $3$ regions; they may overlap because two (and only two)
	of the edges may cross. Nevertheless, the same argument as
	shown before can be used to show that each of the regions defined by 
	the edge $p_{i-1}p_{i+1}$ contains at most $w_i$ points.
	Now a region is bounded by a $3$-edge path, which is defined by $4$ points.

	\begin{figure}[htb]\centering
		\includegraphics[page=15,width=\textwidth]{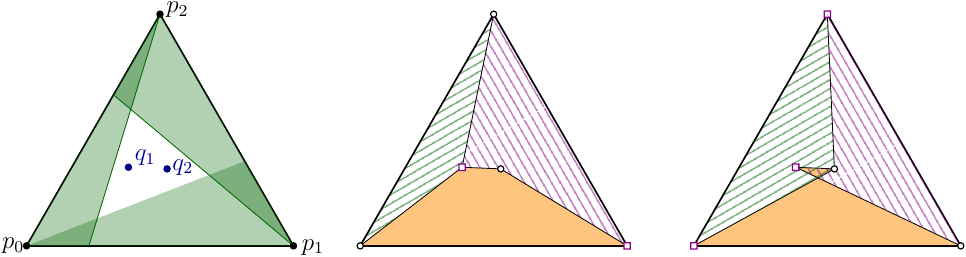}
		\caption{Proof of \cref{lem:separatingpath_colors2}. 
			Left: Two points $q_1,q_2$ from $Q$ with different colors. 
			Center and right: two possible configurations and the regions they define.
			In the right, all three regions share a triangle.}
		\label{fig:separating15}
	\end{figure}
	
	We can now use these regions to cover $CH(P)$ with three pieces, 
	possibly with an overlap, and to select the piece $B$ with the largest
	number of points. Let $\pi$ be the path that together with a portion of 
	the boundary of $CH(P)$ defines $B$.
	As it happened in the proof of \cref{lem:separatingpath_colors1}, 
	$B$ has at least $\frac{n-5}{3}$ points	(instead of $\frac{n-4}{3}$)
	by the pigeonhole principle, and 
	it has at most $w_0\le \left( \frac{1}{3c}+\frac{2}{3}\right) n$ points by construction,
	which means that the complement has at least
	\[
		(n-4)-w_0 ~\ge~ n-4 - \left( \frac{1}{3c}+\frac{2}{3}\right) n = \frac{(c-1)n}{3c}-4
	\]
	points. (Now we have $n-4$, instead of $n-3$, because a bounding path has $4$ points instead of $3$.)
	We take $P'$ to be the points inside $B$ and $P''$ the points outside $B$ and not on $\pi$.
	Since $\pi$ connects two points on the boundary of $CH(P)$, each edge connecting a point
	from $P'$ to a point of $P''$ crosses $\pi$.
\end{proof}

\begin{theorem}\label{thm:separatingpath_colors2}
	Assume that $c\ge 2$ and $n$ is sufficiently large.
	Let $P$ be a set of $n$ points in general position in the plane
	with a balanced $c$-coloring.
	Then $P$ has a polychromatic connected matching of size at least 
	$\frac{(c-1)n}{9c}-\frac{1}{3}$
	which can be computed in $O(n)$ time.
\end{theorem}
\begin{proof}
	The proof is very similar to the proof of \cref{thm:separatingpath_colors1},
	but we use \cref{lem:separatingpath_colors2}.
	We start using \cref{lem:separatingpath_colors2} to obtain the path $\pi$ and 
	the point sets $P'$ and $P''$ claimed there.
	Set $k=\min\{ |P'|,|P''|\} \ge \frac{(c-1)n}{3c}-4$. 
	We construct a polychromatic matching $M$ of size $k$ between $P'$ and $P''$ greedily,
	as discussed in the proof of \cref{thm:separatingpath_colors1}.
	Since each edge of $M$ crosses $\pi$, at least one of the three (or fewer) 
	edges of $\pi$, say $e$, is intersected by $|M|/3$ edges of $M$. 
	The edges of $M$ intersecting $e$
	together with $e$ form a polychromatic matching of size at least
	$1+\frac{k}{3}\ge 1+ \frac{(c-1)n}{9c}-\frac{4}{3} = \frac{(c-1)n}{9c}-\frac{1}{3}$.
\end{proof}

Finally, we compare the bounds of \cref{thm:separatingpath_colors1} 
and \cref{thm:separatingpath_colors2}. For this,
we want to know for which $c$ we have
\[
	\frac{(c-3)n}{6c}-\frac{1}{2} \ge \frac{(c-1)n}{9c}-\frac{1}{3}.
\]
The first bound is better for $c >7$.
(For $c=7$ we get $\frac{2n}{21}-\frac{1}{2}$ against $\frac{2n}{21}-\frac{1}{3}$).

\section{Discussion and Future Work}
\label{sec:discussion}	

We have studied the problem of finding a largest connected matching defined
by a set of points in the plane. Our upper and lower bounds do no match,
and the most obvious open problem is closing the gap.
 
The problem of crossing families asks for
finding a matching in the intersection graph of segments defined
by a set of points. In our problem we were only concerned about connectivity.
A problem in between is the following: 
\begin{question}
	Consider matchings such that the resulting intersection graph is $k$-connected.
\end{question}
For $k=\Theta(n)$, the problem approaches the problem of crossing families.
We can also search for matchings whose intersection graph has additional substructures,
such as containing a largest star, a Hamiltonian path or a Hamiltonian cycle.
Finally, one can consider the algorithmic problem of finding a largest connected
matching (or related structures) for a given point set.

\paragraph*{Acknowledgments.}
Research on this work has been initiated at the 17th European Geometric Graph Week which was held from August 15th to 19th in Leipzig. We thank all participants for the good atmosphere as well as for inspiring discussions on the topic. Moreover, we thank Florian Thomas for pointing us to a variant of Question~\ref{que:main} which initiated our research.

O.~A.~supported by FWF grant~W1230. 
J.~S.~supported by the grant SVV–2023–260699 and by the project 23-04949X of the Czech Science Foundation (GA\v{C}R).

Research funded in part by the Slovenian Research and Innovation Agency (P1-0297, J1-2452, N1-0218, N1-0285).
Research funded in part by the European Union (ERC, KARST, project number 101071836). Views and opinions expressed are however those of the authors only and do not necessarily reflect those of the European Union or the European Research Council. Neither the European Union nor the granting authority can be held responsible for them.

\bibliography{references}
\end{document}